\documentclass[12pt,a4paper]{article}
\usepackage[english]{babel}
\usepackage{pbox}
\usepackage{multirow}
\usepackage{authblk}
\usepackage{algorithm}
\usepackage{amsfonts}
\usepackage{amssymb, amsmath}
\usepackage[utf8]{inputenc}
\usepackage[T1]{fontenc}
\usepackage{amsthm}
\usepackage{verbatim}
\usepackage{listings}
\usepackage{lastpage}

\usepackage{fancyhdr}
\usepackage{amssymb,amsmath}
\usepackage{amsthm}
\usepackage{amsfonts}
\usepackage{algorithm}
\usepackage{listings}
\usepackage[figuresright]{rotating}
\usepackage{authblk}

\newcommand{\network}{\mathcal{N}}
\newcommand{\C}{\mathcal{C}}
\newcommand{\V}{\mathcal{V}}
\newcommand{\lines}{\mathcal{L}}
\newcommand{\indeg}{\mbox{\rm{indeg}}}
\newcommand{\outdeg}{\mbox{\rm{outdeg}}}
\newcommand{\core}{\mbox{\rm{Core}}}
\newcommand{\cores}{\mbox{\rm{Cores}}}
\renewenvironment{proof}{\noindent \textbf{Proof:}\hspace{1mm}}{{\hspace{\stretch{1}}
        \rule{1ex}{1ex}} \vspace{10pt}}
\theoremstyle{definition}
\newtheorem{theorem}{Theorem}[section]
\newtheorem{defin}{Definition}[section]
\newtheorem{corollary}{Corollary}[section]

\newtheorem{lemma}{Lemma}[section]

\title{Generalized Two-mode Cores}
\author[1]{Monika Cerin\v sek}
\author[2,3]{Vladimir Batagelj}
\affil[1]{Hru\v ska d.o.o., Kajuhova 90, 1000 Ljubljana}
\affil[2]{Institute of Mathematics, Physics and Mechanics, Jadranska 19, 1000 Ljubljana}
\affil[3]{University of Ljubljana, FMF, Department of Mathematics, Jadranska 19, 1000 Ljubljana}
\date{}

\begin{document}
\maketitle

\begin{abstract}
The node set of a two-mode network consists of two disjoint subsets and all its links are linking these two subsets. The links can be weighted. We develo\-ped a new method for identifying important subnetworks in two-mode networks. The method combines and extends the ideas from generalized cores in one-mode networks and from $(p,q)$-cores for two-mode networks.
In this paper we introduce the notion of generalized two-mode cores and discuss some of their properties. An efficient algorithm to determine generalized two-mode cores and an analysis of its complexity are also presented. 
 For illustration some results obtained in analyses of real-life data are presented.
\end{abstract}




% \linenumbers

\section{Introduction}
\label{intro}
Network analysis is an approach to the analysis of relational data. In this paper we deal with the analysis of two-mode networks~\cite{borgatti,latapy}. A two-mode network is a network in which the set of nodes consists of two disjoint subsets and its links are linking these two subsets.

The traditional approach to the analysis of two-mode networks is usually indirect: first
a two-mode network is converted into one of the two corresponding one-mode projections, and afterward it is analyzed using standard network analysis methods \cite{wasserman}. \emph{Direct} methods for the analysis of two-mode networks are quite rare \cite{sn, largescale, sage}. We can use bipartite statistics on degrees, generalized blockmodeling, $(p,q)$-cores, two-mode hubs and authorities, $4$-ring weights, bi-communities, two-mode clustering, bipartite cores, and some others. Many methods for identifying important subnetworks are available for one-mode networks (measures of centrality and importance, generalized cores, line islands, node islands, clustering, blockmodeling, etc.).
We present a new direct method which can be used for the identification of important subnetworks in two-mode networks with respect to selected node properties.

We combine the ideas from generalized cores in one-mode networks and from $(p,q)$-cores for two-mode networks into the notion of \emph{generalized two-mode cores}. 
We developed and implemented an algorithm for identifying generalized two-mode cores for selected node properties and given thresholds for both subsets of nodes.
We also propose an algorithm to find the nested generalized two-mode cores for one fixed threshold value.

In the next section we survey the works that contain the ideas we used for the development of our method. In Section~\ref{methods} we present an algorithm for identifying the generalized two-mode cores. We list some node properties that are used as measures of importance. We also present some properties of generalized two-mode cores. We prove that for equivalent properties measured in ordinal scales the sets of generalized two-mode cores are the same. The algorithm, the proof  of its correctness, and a simple analysis of its complexity are presented in Section~\ref{algo}. In Section~\ref{rez} some results obtained in analyses of real-life data are presented.

\section{Related work}
\label{related}
The notion of $k$-core was introduced by Seidman (1983)~\cite{seidman}.
Let $\mathcal{G} = (\V, \lines)$ be a graph with $n = |\V|$ nodes and $m = |\lines|$ links. 
Let $k$ be a fixed integer and let $\deg(v)$ be the degree of a node $v \in \V$. 
A subgraph $\mathcal{H}_k = (\C_k, \lines|\C_k)$ induced by the subset $\C_k \subseteq \V$ is called a \textit{k-core} iff $\deg_{\mathcal{H}_k} (v) \geq k$, for all $v \in \C_k$, and $\mathcal{H}_k$ is the maximal such subgraph.
If we replace the degree with some other node property, we get the notion of \textit{generalized cores} as it was introduced in~\cite{zaversnik}. The node property can be a node degree, maximum of incident link weights, sum of incident link weights, etc. They are described in more details in Section~\ref{methods}.

The other possible generalization of $k$-cores is their extension to two-mode networks. 
The notion of  $(p,q)$-cores was introduced in~\cite{imdb}. 
A subset $\C \subseteq \V$ determines a $(p,q)$-core in a two-mode network $\network = ((\V_1, \V_2), \lines), \V = \V_1 \cup \V_2$ iff in the subnetwork $\mathcal{K} = ((\C_1, \C_2), \lines|\C), \C_1 = \C \cap \V_1, \C_2 = \C \cap \V_2$ induced by $\C$ it holds that for all $v \in \C_1: \deg_{\mathcal{K}}(v) \geq p$ and for all $v \in \C_2: \deg_\mathcal{K}(v) \geq q$, and $\C$ is the maximal such subset in $\V$.

We combined the ideas from generalized cores and $(p,q)$-cores into the notion of generalized two-mode cores.
Generalized two-mode cores are defined similarly to $(p,q)$-cores with one exception -- instead of using the degree of nodes, we now allow also some other properties of nodes. 
The properties of nodes on subsets $\V_1$ and $\V_2$ can be different.
The detailed definition is given in Section \ref{prop}.

Other types of two-mode subnetworks were discussed in the literature. A bipartite core is defined as a complete two-mode subnetwork \cite{bicores}. Bipartite cores are determined by the size of each subset of vertices. %Bipartite cores are actually $(p,q)$-cores with values of parameters $p$ and $q$ equal to maximal node degrees in both subsets of nodes of a network.

Similar methods are varieties of a community detection in two-mode networks: with maximization of monotone function \cite{col0}, with a dual projection \cite{com10}, with properties of the eigenspectrum of the network's matrix \cite{com1}, with label propagation and recursive division of the two types of nodes \cite{com9}, with the stochastic block modeling \cite{com5}, and many others. Another very similar method is a bipartite clustering \cite{clu2, clu1}. The substantial difference is that these methods are determining a clustering of the whole set of nodes and our method determines only an important subset.

The generalized two-mode cores depend on selected node properties that are expressing different aspects of the network structure (for example, the intensity of links). They are also using different criteria. Therefore our method represents a new approach to two-mode network analysis. 
It does not represent an improvement of any existing method, but a generalization of $(p,q)$-cores.

\section{Algorithms for generalized two-mode cores}
\label{methods}
As mentioned in Section~\ref{intro} the algorithms for identifying $k$-cores, generalized cores, and $(p,q)$-cores have already been developed \cite{imdb,zaversnik,seidman}. We propose a new algorithm, which combines and extends the ideas from generalized cores in one-mode networks and from $(p,q)$-cores in two-mode networks. Besides implementing the new algorithm, we also prove its correctness and analyze its complexity.
For testing the usefulness of the method we applied it on real networks.

\subsection{Properties of nodes}
\label{prop}
For a network $\network = (\V, \lines, w)$ and a weight function $w: \lines \rightarrow \mathbb{R}^+$ a \emph{property function} $f(v, \C)  \in \mathbb{R}_0^+$ is defined for all $v \in \V$ and $\C \subseteq \V$.
A subset $\C$ induces the subnetwork to which the evaluation of the property function is limited.
In an undirected network it holds: $w(u,v) = w(v,u)$ for all pairs of nodes $u,v \in \V.$

Let us denote the neighborhood of a node $v$ as $N(v)$ and the neighborhood of a node $v$ within the subset $\C$ as $N(v, \C) = N(v)\cap \C$. The neighborhood of a node $v$ within the subset $\C$ including $v$ we denote as $N^+(v, \C) = N(v, \C) \cup \{v\}.$ Let us also denote a measurement on nodes (degree, centrality, etc.) as $t: \V \rightarrow \mathbb{R}^+_0.$

We say that the property function
$f(v, \C)$ is \emph{local} iff
$$
f(v, \C) = f(v, N(v, \C))\qquad\textrm{for all }v \in \V \textrm{ and } \C \subseteq \V.
$$
The property function $f(v, \C)$ is \emph{monotonic} iff
$$
\C_1 \subset \C_2 \implies \forall v \in \V: f(v, \C_1) \leq f(v, \C_2).
$$

Some node properties ($f_1$ -- $f_{10}$) were proposed in~\cite{zaversnik}.
In the Tab.~\ref{tab_props} are listed examples of property functions.
\begin{table}
\caption{Examples of property functions.}\label{tab_props}
\begin{tabular}{p{8.3cm}p{7.5cm}}
\textbf{Formula} & \textbf{Description}\\
\hline
$f_1 (v, \C) = \deg_\C(v)$ 			& Degree of a node $v$ within $\C$.\\
$f_2 (v, \C) = \indeg_\C(v)$ 		& Input degree of a node $v$ within $\C$.\\
$f_3 (v, \C) = \outdeg_\C(v)$ 		& Output degree of a node $v$ within $\C$.\\
$f_4 (v, \C) = \indeg_\C(v) + \outdeg_\C(v)$& For a directed network $f_4 = f_1$.\\
$f_5 (v, \C) = \mbox{\rm{wdeg}}_\C(v) = \sum_{u \in N(v, \C)} w(v,u)$ & Sum of weights of links within $\C$ that have a node $v$ for an end node.\\
$f_6 (v, \C) = \mbox{\rm{mweight}}_\C(v) = \max_{u \in N(v, \C)} w(v, u)$ & Maximum of weights of all links within $\C$ that have a node $v$ for an end node.\\
$f_7 (v, \C) = \mbox{\rm{pdeg}}_\C(v) =  \frac{\deg_\C(v)}{\deg(v)}$ if $\deg(v) > 0$ else $f_7 (v, \C) = 0$ & Proportion of $N(v, \C)$ in $N(v)$.\\
$f_8 (v, \C) = \mbox{\rm{density}}_\C(v) = \frac{\deg_\C(v)}{\max_{u \in N(v)} \deg(u)}$ if $\deg(v) > 0$ else $f_8(v, \C) = 0$ & Relative density of the neighborhood of a node $v$ within $\C$.\\
\parbox{7cm}{$f_{9} (v, \C) = \mbox{\rm{degrange}}_\C(v) = $\\ $\max_{u \in N(v, \C)} \deg (u) - \min_{u \in N(v, \C)} \deg (u)$}
 & Range of degrees of neighbors of a node $v$ within $\C$ with respect to their degrees.\\
\parbox{7cm}{$f_{10} (v, \C) = \mbox{\rm{tdegrange}}_\C(v) = $\\$\max_{u \in N^+(v, \C)} \deg (u) - \min_{u \in N^+(v, \C)} \deg (u)$} & Total range of degrees of neighbors of a node $v$.\\
$f_{11} (v, \C) = \mbox{\rm{pweight}}_\C(v) = \frac{\sum_{u \in N(v, \C)} w(v,u)}{\sum_{u \in N(v)} w(v,u)}$ if $\sum_{u \in N(v)} w(v,u) > 0$ else $f_{11} (v, \C) = 0$ & Proportion of weights of links with a node $v$ as an end node that have the other end node within $\C$.\\
$f_{12} (v, \C) = \mbox{\rm{triangles}}_\C(v)$ & Number of triangles through a node $v$ with all three nodes in $\C$.\\
$f_{13} (v, \C) = \mbox{sum } \C(v,t) = \sum_{u \in N(v, \C)}t(u).$&\\
$f_{14} (v, \C) = \max \C(v,t) = \max_{u \in N(v, \C)}t(u).$&
\end{tabular}
\end{table}

All the listed functions have the property $f(v, \emptyset) = 0$ for all $v \in \V$.

It can easily be verified that all the listed property functions are local and monotonic. An example of a non-monotonic function would be the average weight
$$
f(v, \C) = \frac{1}{\deg_{\C}(v)}\displaystyle\sum_{u \in N(v, \C)} w(v, u)
$$
for $N(v, \C) \neq \emptyset$, otherwise $f(v, \C) = 0$. An example of a non-local function is the number of cycles or closed walks of length $k$, $k \geq 4$, through a node.

\subsection{Generalized two-mode cores}
\label{gtc}
\begin{defin}\label{defgtcore}
Let $\network = ((\V_1, \V_2), \lines, (f, g), w), \V = \V_1 \cup \V_2 $ be a finite two-mode network -- the sets $\V$ and $\lines$ are finite. Let $\mathcal{P}(\V)$ be a power set of the set $\V$.  Let functions $f$ and $g$ be defined on the network $\network$: $f, g: \V \times \mathcal{P}(\V) \longrightarrow \mathbb{R}_0^+$.

A subset of nodes $\C \subseteq \V$ in a two-mode network $\network$ is a \textit{generalized two-mode core} $\C = \core(p,q; f,g),\; p,q \in \mathbb{R}_0^+$ if and only if in the subnetwork $\mathcal{K} = ((\C_1, \C_2), \lines|\C), \C_1 = \C \cap \V_1, \C_2 = \C \cap \V_2$ induced by $\C$ it holds that for all $v \in \C_1: f(v,\C) \geq p$ and for all $v \in \C_2: g(v, \C) \geq q$, and $\C$ is the maximal such subset in $\V$.
\end{defin}

When the functions $f$ and $g$ are clear from context and the parameters $p$ and $q$ are fixed, we use the abbreviation $\core(p,q) \equiv \core(p,q; f,g)$.

A set of all generalized two-mode cores for a network $\network = ((\V_1, \V_2), \lines, (f,g), w)$ is defined as
$$
\cores(\network) = \{ \core(p,q; f,g); p \in \mathbb{R}^+_0 \land q \in \mathbb{R}^+_0 \}.
$$
In this set we can define a relation
$$
\core(p,q;f,g) \sqsubseteq \core(p',q';f,g)
$$
as
$$
\C_1(p,q;f,g) \subseteq \C_1(p',q';f,g) \, \land \, \C_2(p,q;f,g) \subseteq \C_2(p',q'; f,g).
$$

Because  $\cores(\network) \subseteq \mathcal{P}(\V)$ and $(\mathcal{P}(\V), \subseteq)$ is a partially ordered set, also $(\cores(\network), \subseteq)$ and $(\cores(\network), \sqsubseteq)$ are partially ordered.

For the generalized two-mode cores we have:
\begin{itemize}
\item $\core(0,0) = \V$,
\item the subnetwork $\mathcal{K}= ((\C_1, \C_2), \lines|\C),\; \C = \core (p,q;f,g)$ is not necessarily connected.
\end{itemize}

\begin{lemma}\label{lem1}
Let $\network = ((\V_1, \V_2), \lines, (f,g), w)$ and its "mirror" $\tilde{\network} = ((\V_2, \V_1), \lines, (g,f), w)$ be two-mode networks.
It holds 
$$
\core_{\tilde{\network}} (p,q; f,g) = \core_{\network}(q,p;g,f).
$$
\end{lemma}

\begin{lemma}\label{lem2}
Let $F \subseteq\mathbb{R}$ be a set of values of the property $f$ (codomain of the property function $f$) and $\varphi : F \longrightarrow \mathbb{R}_0^+$ a strictly increasing function. Then
$$
\core(p,q; f,g) = \core(\varphi(p),q; \varphi \circ f, g).
$$
\end{lemma}

\begin{corollary}\label{cor1}
Let $F \subseteq\mathbb{R}$ and $\varphi : F \longrightarrow \mathbb{R}_0^+$ be as in Lemma \ref{lem2}. Then
$$
\core(q,p;g,f) = \core(q,\varphi(p); g, \varphi \circ f).
$$
\end{corollary}

The proofs for Lemmas \ref{lem1} and \ref{lem2} and Corollary \ref{cor1} are simple. Lemma \ref{lem2} and Corollary \ref{cor1} tell us that equivalent properties measured in ordinal scales produce the same generalized two-mode cores.

\subsection{Boundary for threshold values}
\label{boundsec}
\begin{theorem}
In a two-mode network $\network = ((\V_1, \V_2), \lines, (f,g),w)$ for $f$ and $g$ monotonic functions it holds:
$$
(p_1 \leq p_2 \, \land \, q_1 \leq q_2) \implies \core(p_2, q_2) \subseteq \core(p_1, q_1).
$$
\end{theorem}

\begin{proof}
By definition of a generalized two-mode core $\C = \core(p_2, q_2)$ it holds:
$$
\forall v \in \C_1: f(v, \C) \geq p_2  \quad\textrm{ and }\quad  \forall v \in \C_2: g(v, \C) \geq q_2
$$
and $\C$ is maximal such subset of nodes.
Because $f$ and $g$ are monotonic and we have $p_1 \leq p_2$ and $q_1 \leq q_2$ it also holds:
$$ 
\forall v \in \C_1: f(v, \C) \geq p_1  \quad\textrm{ and }\quad  \forall v \in \C_2: g(v, \C) \geq q_1.
$$
But $p_1 \leq p_2$ and $q_1 \leq q_2$ so $\C$ is not necessarily the maximal subset of nodes that defines $\core(p_1, q_1).$
Therefore:
$$
\core (p_2, q_2) = \C \subseteq \core (p_1, q_1).
$$
\end{proof}

For a given subset $\C \subseteq \V$ let $p(\C) = \min_{v \in \C_1} f(v, \C)$ be the minimum property value in the set $\C_1 = \C \cap \V_1$ and $q(\C) = \min_{v \in \C_2} g(v, \C)$ the minimum property value in the set $\C_2 = \C \cap \V_2$. It holds $\C \subseteq \core(p(\C), q(\C))$.

For a given two-mode network $\network = (\V, \lines, (f,g),w) $ let $P = \{ p(\C): \C \subseteq \V \}$ be the set of all possible values of $p(\C)$ and $Q = \{ q(\C): \C \subseteq \V \}$ the set of all possible values of $q(\C)$. $P$ and $Q$ are finite sets. Therefore we can enumerate their elements:
$$
\begin{array}{c}
P = \{ p_1, p_2, \ldots, p_r \},\quad p_i < p_{i+1},\\
Q = \{ q_1, q_2, \ldots, q_s \},\quad q_i < q_{i+1}.
\end{array}
$$

For fixed functions $f$ and $g$ we are interested only in $(p,q)$ pairs that are determining different nonempty generalized two-mode cores.

It is clear from the condition $p,q \geq 0$ that if we look at the Cartesian coordinate system with $(p, q)$ axes (Fig.~\ref{pq}), the region of all possible pairs of thresholds $p$ and $q$ is limited to the first quadrant. 
The missing boundary of this region is a broken line in a shape of stairs.
\begin{figure}[!h]
\centering
\includegraphics[width=6.6cm]{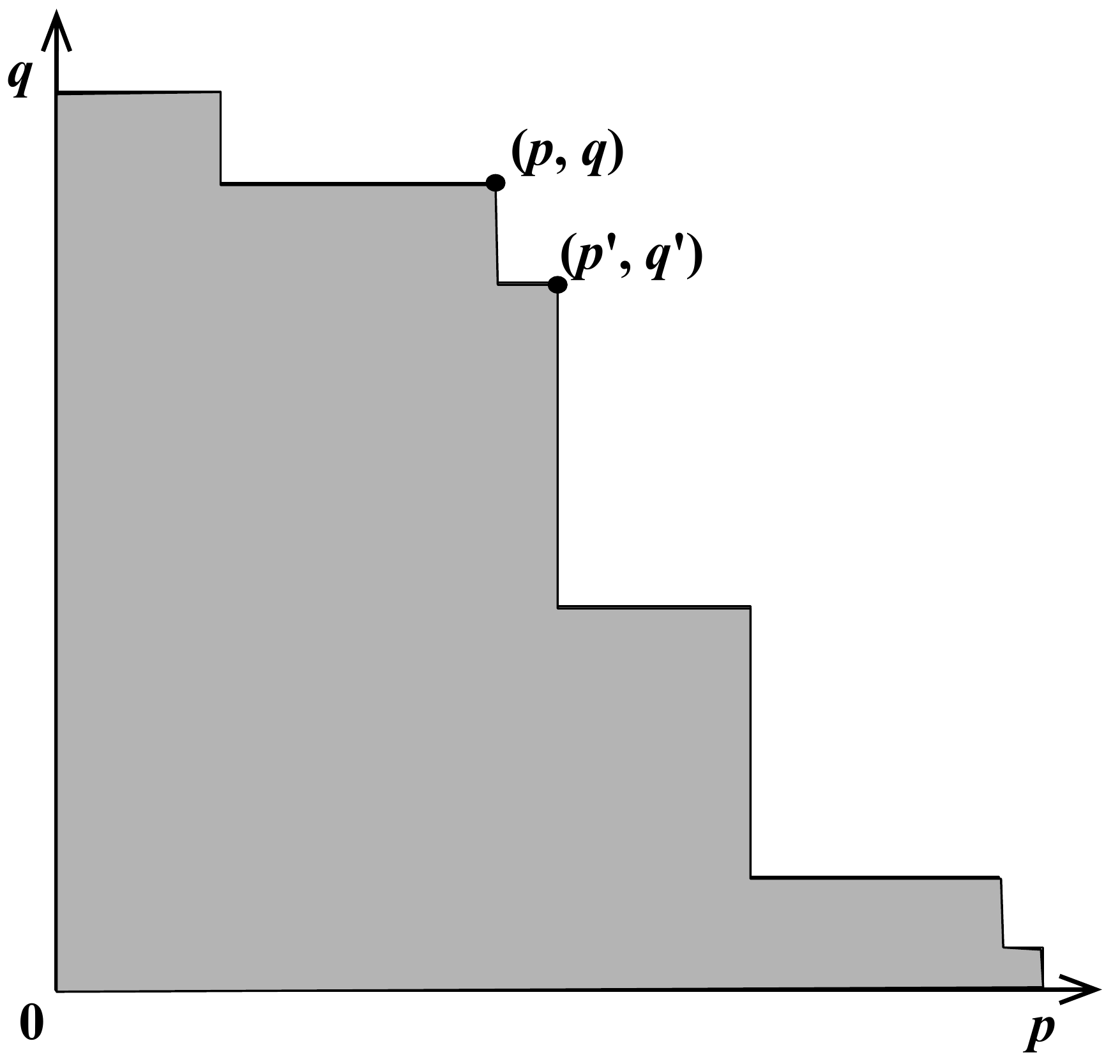}
\caption{The region of all possible threshold pairs $(p,q)$ that determine generalized two-mode cores.}
\label{pq}
\end{figure}

For $p, q \in \mathbb{R}_0^+$ let $\Pi(p) = \{ \C: \C \subseteq \V \land p(\C) \geq p \}$ be the set of sets for which the property value of the first subset is at least equal to $p$, and $\Gamma(q) = \{ \C: \C \subseteq \V \land q(\C) \geq q \}$ be the set of sets for which the property value of the second subset is at least equal to $q$.  Let $G(p) = \{ q(\C): \C \in \Pi(p) \}$ be the set of all possible values $q(\C)$ for which the set $\C$ belongs to the set $\Pi(p)$ and $q_{\Pi(p)} = \max G(p)$ is the maximum such value $q(\C)$. Let $H(q) = \{ p(\C): \C \in \Gamma(q) \}$ be the set of all possible values $p(\C)$ for which the set $\C$ belongs to the set $\Gamma (q)$ and $p_{\Gamma(q)} = \max H(q)$ is the maximum such value $p(\C)$.

\begin{lemma}
The following properties hold for $p, p' \in P$ and $q, q' \in Q$:
\begin{enumerate}
\item In the set $\Gamma (q_{\Pi(p)})$ exists at least a set $\C$ for which it holds $p(\C) = p_0$ and $q(\C) = q_{\Pi(p_0)}$. Therefore:
$$\Gamma (q_{\Pi(p)}) \neq \emptyset.$$ Similarly: $\Pi (p_{\Gamma(q)}) \neq \emptyset$.
\item It holds $\C \subseteq \core(p(\C), q(\C))$ and in all sets in $\Gamma (q_{\Pi(p)})$ the property values for the second subset are at least equal to $q_{\Pi(p)}$. It also  holds:
$$\C \in \Gamma (q_{\Pi(p)}) \implies \C \subseteq \core(p, q_{\Pi(p)}).$$ Similarly: $\C \in \Pi (p_{\Gamma(q)}) \implies \C \subseteq \core(p_{\Gamma(q)}, q)$.
\item $q_{\Pi(p)}$ is the maximum such $q$ for which a nonempty core $\core(p,q)$ exists:
$$q > q_{\Pi(p)} \implies \core(p, q) = \emptyset.$$ Similarly: $p > p_{\Gamma(q)} \implies \core(p, q) = \emptyset.$
\item $q_{\Pi}$ is a decreasing function:
$$p < p' \implies q_{\Pi(p')} \leq q_{\Pi(p)}.$$ Similarly: $q < q' \implies p_{\Gamma(q')} \leq p_{\Gamma(q)}.$
\end{enumerate}
\end{lemma}

\begin{proof}
The proofs for all of the listed properties are simple. Let us show just the first part of the property 4.
$$
\begin{array}{rcl}
\Pi(p')  & = & \{ \C:   \C \subseteq \V \land p(\C) \geq p' \}\\
& \subseteq & \{ \C:   \C \subseteq \V \land p(\C) \geq p \}  =  \Pi(p) \\
\end{array}
$$
This implies 
$
G(p') \subseteq G(p)
$
 and therefore
$q_{\Pi(p')} \leq q_{\Pi(p)}$.
\end{proof}

From these properties it follows that for each $p \in P$ there exist a maximum $q \in Q$: $q = q_{\Pi(p)}$; and for each $q \in Q$ there exists a maximum $p \in P$: $p = p_{\Gamma(q)}$. This is a formal proof of the staircase shape of the boundary. On this basis we can develop an algorithm for determining the boundary of the set $(P, Q)$ in the $(p,q)$ coordinate system -- see  Algorithm \ref{bound}.
\begin{algorithm}
\caption{The algorithm to determine the boundary of the set $(P, Q)$ in the $(p,q)$ coordinate system.}
\label{bound}
\begin{lstlisting}[mathescape]
Input:  $P = \{ p_1, p_2, \ldots, p_r \}, p_i < p_{i+1},$ $Q = \{ q_1, q_2, \ldots, q_s \}, q_i < q_{i+1}$.
Output: boundary set $boundary \subseteq P \times Q$.

Algorithm:
$q_{\max} = 0$
$boundary = \emptyset$

for $p \in reverse(P)$ do
     $q = q_{\Pi}(p)$
     if $q > q_{max}$ then
          $q_{max} = q$
          $boundary = boundary \cup \{ (p,q) \}$
\end{lstlisting}
\end{algorithm}

In general the Algorithm \ref{bound} is only of theoretical value because the sizes of sets $P$ and $Q$ can be very large. It can be used for some special property functions -- for example $f_1, f_2, f_3$ and $f_4$, where the sets $P$ and $Q$ are relatively small.

\section{The algorithm}
\label{algo}
We propose an algorithm for  determining a generalized two-mode core in two-mode networks for given thresholds $p$ and $q$, and property functions $f$ and $g$.

The basic idea of the algorithm for generalized two-mode core is to repeat 
removing the nodes that do not belong to it. Since the network is two-mode the
removing condition depends on to which set a node belongs.

\begin{algorithm}
\caption{Basic algorithm for determining a generalized two-mode core.}
\label{basic}
$\mathcal{C} = \mathcal{V}$\\
\textbf{while} $\exists v \in \mathcal{C} : 
(v \in \mathcal{C}_1 \land f(v,\mathcal{C}) < p) \lor
(v \in \mathcal{C}_2 \land g(v,\mathcal{C}) < q)$ \textbf{do}\\
\hspace*{10mm} $\mathcal{C} = \mathcal{C} \setminus \{ v \}$

\end{algorithm}

Adapting the proof of Theorem 1 from \cite{zaversnik} to two-mode networks we can prove the following theorem.

\begin{theorem}
The Algorithm~\ref{basic} determines the generalized two-mode core at thresholds $(p, q)$ for every monotonic node property functions $f(v, \C)$ and $g(v, \C)$. The result of the algorithm does not depend on the order of deletion of nodes that do not belong to the core -- satisfy the while loop condition.
\end{theorem}

\begin{algorithm}
\caption{The algorithm to determine the generalized two-mode core at thresholds $(p,q)$ for monotonic and local node property functions $f$ and $g$.}
\label{alg}
\begin{lstlisting}[mathescape]

Input: two-mode network $\network = ((\V_1, \V_2), \lines, (f,g), w)$, $\V_1 \cap \V_2 = \emptyset$,
        thresholds $p$, $q$.
Output: generalized two-mode core $\C = \C_1 \cup \C_2$.

Procedure:
Remove($h$, $t$, $\C_{current}$, $\C_{other}$, heap$_{current}$, heap$_{other}$):
     while not Empty(heap$_{current}$):
          if Root(heap$_{current}$).$value$ $\geq t$: return
          $u$ = RemoveRoot(heap$_{current}$).$key$
          $\C_{current}$ = $\C_{current} \backslash \{u\}$
          for $v$ in $N(u, \C_{other})$: 
                update(heap$_{other},v, h(v, \C))$

Algorithm:
$\C_1$ = $\V_1$, $\C_2$ = $\V_2$

for $v$ in $\V_1$:   value[$v$] = $f(v, \V)$
for $v$ in $\V_2$:   value[$v$] = $g(v, \V)$

build(heap$_1$,value, $\C_1$), build(heap$_2$,value, $\C_2$)

repeat:
     Remove($f$, $p$, $\C_1$, $\C_2$, heap$_1$, heap$_2$)
     Remove($g$, $q$, $\C_2$, $\C_1$, heap$_2$, heap$_1$)
until no node was removed
\end{lstlisting}
\end{algorithm}

Because the order of deletions has no impact on the final result we can, in our elaboration of the algorithm, start deleting from the first set, then switch to the second set, and back to the first set, etc., until no deletion is possible.

We use a binary heap implementation of priority queues~\cite{algoritmi} to organize the nodes in such a way that we can efficiently get the node with the smallest value of the property function as the root element in a heap. The function \verb$RemoveRoot(heap)$ returns the root element and removes it from the heap.
The value of the property function is calculated for each node according to the set of links $\lines$ and the weight function $w$ as described in Section~\ref{prop}.

We need two heaps -- one for each subset of nodes.
An element $E$ in a heap is a pair $E = (key, value)$ of an identificator of a node as $E.key$ and a property function value as $E.value$.
The elements are ordered by the value of the property function. 
We know that all neighbors of every node in the first subset are in the second subset,
and vice versa.

To bound the work needed for updating we shall assume in the following that both functions $f$ and $g$ are local. In this case we need to update only the heap of the neighboring subset of nodes when deleting a node. 

These decisions result in  Algorithm~\ref{alg} for determining generalized 
$(p,q)$-core for monotonic and local property functions $f$ and $g$.

%\begin{corollary}
%If the property functions $f(v, \C)$ and $g(v, \C)$ are monotonic it holds:
%$$
%p_1 \leq p_2 \land q_1 \leq q_2 \implies \core(p_2, q_2; f,g) \subseteq \core(p_1, q_1; f,g).
%$$
%\end{corollary}

%\begin{proof}
%The order of the deletion of nodes during the execution of the Algorithm \ref{alg} is not affecting the result. We first determine the generalized two-mode core $\core(p_1, q_1; f,g)$. Afterward we increase the thresholds to values $p_2$ and $q_2$. Because of the increase maybe some additional nodes are deleted and we get the $\core(p_2, q_2; f,g)$.
%\end{proof}

\subsection{Complexity}
\label{complexity}
Algorithm \ref{alg} could be based also on a simpler data structure such as queues. We did not explore these options because they can not be efficiently extended to Algorithm \ref{alg1}.

Assume that the time complexity of the calculation of a local property function is $O(\deg (v))$ for each node $v$ in the network.
The building of the heap takes $O(n \cdot \log n)$. We use binary heaps instead of faster heaps with amortized constant update time complexity because they are easier to implement.
Because $\sum_{v \in \V} \deg (v) = 2m$, the time complexity of the initialization is 
$$
O(\displaystyle \sum_{v \in \V} \deg (v) + n \log n) = O(\max (m, n \log n)) .
$$

At each step of the while loop some node is removed and its neighbors get their property function value changed. They also change their position in the heap according to the change of their value.

The update of the property function value may require less time than its calculation. For example the value of $f_1 (v, \C) = \deg_{\C} (v)$ is corrected only by reducing its value by one for every removed neighbor of node $v$. The update of the value of property functions $f_1, f_2, f_3, f_4, f_5, f_7, f_8, f_{11}$, and $f_{13}$ takes $O(1)$ time; but for property functions $f_6, f_9, f_{10}, f_{12}$, and $f_{14}$ it takes $O(\deg (v))$ for every node $v.$

The heaps are implemented in such a way that the change of the position of an element in it takes $O(\log  n )$ time.

Let $s = (v_1, v_2, \ldots, v_d)$ be the sequence of the removed nodes during the execution of the algorithm. The removal of a node $v_i$ takes $O(\log  n )$. The update of the property function value of each of its neighbors and the change of its position in the heap takes $O(\log  n )$ or $O(\max (\log  n, \deg (v_i)))$ depending on the property function used, and for the sequence $s$ the update costs
$$
\displaystyle\sum_{v_i \in s} \deg (v_i) \cdot O(\log n) \leq 2m \cdot  O(\log n) =  O(m \cdot \log n)
$$
for property functions $f_1, f_2, f_3, f_4, f_5, f_7, f_8, f_{11}$, and $f_{13}$ or
$$
\begin{array}{cl}
& \displaystyle\sum_{v_i \in s} \deg (v_i) \cdot O(\max (\log n, \deg (v_i))) \\ 
\leq & O(\max (\log n, \Delta)) \cdot \displaystyle\sum_{v_i \in s} \deg (v_i) \\
\leq & 2m \cdot O(\max (\log n, \Delta)) =  O(m \cdot \max (\log n, \Delta)),
\end{array}
$$
where $\Delta = \max_{v \in \V} (\deg (v))$,
for property functions $f_6, f_9, f_{10}, f_{12}$, and $f_{14}$.

The time complexity of the initialization of the algorithm is lower than the time complexity of the main loop in the algorithm. So the time complexity of the whole algorithm for determining the generalized two-mode core is
$$
O(m \cdot \log n)
$$
for property functions $f_1, f_2, f_3, f_4, f_5, f_7, f_8, f_{11}$, and $f_{13}$ and 
$$
O(m \cdot \max (\log n, \Delta))
$$
for property functions $f_6, f_9, f_{10}, f_{12}$, and $f_{14}$.

\subsection{An algorithm for one threshold value fixed}
\label{generalization}
We adapted the algorithm so that one subset of nodes has a fixed threshold value.
The result of this algorithm is a vector in which every node has as its value the maximum value of the nonfixed threshold for which it is still in the corresponding generalized two-mode core.
This helps selecting the thresholds for which we get the most important generalized two-mode cores.

Algorithm \ref{alg1} presents such an adapted version of the Algorithm \ref{alg}. In this algorithm we fix the first threshold. In the case where the second threshold is fixed we apply the Lemma \ref{lem1}.

The algorithm is again based on the idea of deleting the nodes that do not satisfy the threshold. 
In the elaboration of this algorithm we also use heaps.
The value of the property function is calculated for each node according to the set of links $\lines$ and the weight function $w$ as described in Section~\ref{prop}.
The nodes are ordered in heaps for both subsets of nodes by the value of the property functions.

A step of the while loop starts by deleting all nodes in the heap for the first subset of nodes that do not satisy the fixed threshold. 
Removed nodes get the value of the second threshold in the previous step of the loop, because this is the largest value of the second threshold for which the removed node is still in the generalized two-mode core.
Property function values of some neighboring nodes might be changed, so we set the second threshold $q$ to the current smallest value in the heap for the second subset of nodes after the first part of the step.
Then we remove all nodes from the second heap that have the property function value equal to the current $q$. Removed nodes get the value $q$ in the resulting vector.
At the end of the step we set the  previous value of $q$ to its current value and its current value to the value of the root of the second heap.

The complexity of Algorithm \ref{alg1} is the same as the complexity of Algorithm \ref{alg}.

\begin{algorithm}
\caption{The algorithm to determine the vector of generalized two-mode core levels at fixed threshold $p$ for monotonic and local node property functions $f$ and $g$.}
\label{alg1}
\begin{lstlisting}[mathescape]
Input: two-mode network $\network = ((\V_1, \V_2), \lines, (f,g), w)$, $\V_1 \cap \V_2 = \emptyset$,
     threshold $p$.
Output: vector $T$, $T[v] = $ max value of the treshold $q$ such 
     that $v \in \core(p,q)$.

Procedures:
RemoveFixed($f, p, q, \C_{1}, \C_{2},$ heap$_{1},$ heap$_{2}$):
     while not Empty(heap$_1$):
          if Root(heap$_1$).$value$ $\geq p$: return
          $u$ = RemoveRoot(heap$_1$).$key$
          $T[u] = q$
          $\C_{1} = \C_{1} \backslash \{ u \}$ 
          for $v$ in $N(u, \C_{2}):$
               update(heap$_{2},v, f(v, \C_{1}))$

RemoveChanging($g, q, \C_{2}, \C_{1}$, heap$_{2}$, heap$_{1}$):
     while not Empty(heap$_2$):
          if Root(heap$_2$).$value$ $> q$: return
          $u$ = RemoveRoot(heap$_2$).$key$
          $T[u] = q$
          $\C_{2} = \C_{2} \backslash \{ u \}$ 
          for $v$ in $N(u, \C_{1}):$
               update(heap$_{1},v, g(v, \C_{2}))$

Algorithm:
$\C_1 = \V_1, \C_2 = \V_2$
for $v$ in $\V_1:$   value[$v$] = $f(v, \V_2), \, T[v] = -1$
for $v$ in $\V_2:$   value[$v$] = $g(v, \V_1), \, T[v] = -1$
$q = -1$
build(heap$_1$,value, $\V_1$), build(heap$_2$,value, $\V_2$)
repeat:
     RemoveFixed($f, p, q, \C_1, \C_2$, heap$_1$, heap$_2$)
     if not Empty(heap$_2$):
          q = Root(heap$_2$).value
          RemoveChanging($g, q, \C_2, \C_1$, heap$_2$, heap$_1$)
until Empty(heap$_1$) $\land$ Empty(heap$_2$)
\end{lstlisting}
\end{algorithm}

\section{Applications}
\label{rez}
The possibility of using different node properties for the identification of important two-mode subnetworks allows the analysts to gain information about different types of groups with a single method. We are continuing to search for node properties to include them into our list and the supporting program and apply them on real-life data.

The method of generalized two-mode cores can be applied to  different real-life data. The input network for the method does not need to be a two-mode network. The method can also be applied to an  one-mode network considered as a two-mode network. For example, in an authors' citation network the 'row'-authors can be considered as users, and the 'column'-authors as (knowledge) providers. The use of generalized two-mode cores on a one-mode network allows the analyst to find a subnetwork that is characterized by two different property functions.

\subsection{Social networks}
\label{sn}
Let us take a look at an example. Web of Science is a bibliographic database. We used the data
obtained in 2008 from this database for a query "social network*" and expanded with descriptions of most frequent references and bibliographies of around 100 social networkers. We constructed some two-mode networks. Two among them are also the networks works $\times$ journals and works $\times$ authors ($193376$ works, $14651$ journals, and $75930$ authors). We multiply the transpose of the first network with the second network to get the network journals $\times$ authors.
A journal and an author are linked if the author published at least one work in that journal. The weight of a link is equal to the number of such works.

The simplest generalized two-mode cores are the ones with both property functions the same. If we select
$$f_A(v, \C) = g_A(v, \C) = \deg_{\C}(v)$$
 we get the ordinary $(p,q)$-cores. In a $(p,q)$-core are the journals that published works of at least $p$ different authors in this core and the authors that published their works in at least $q$ different journals in this core.
 
 We determined the generalized two-mode core for functions $f_A$ and $g_A$
 and selected threshold values $p=85$ and $q=3$ for which we get the smallest two-mode core. This is one of the generalized two-mode cores on the border of $(p,q)$ region. It determines a subnetwork of journals in which at least $85$ authors (in this subnetwork) published their works, and of authors that published their works in at least $3$ different journals (in this subnetwork). There are $4$ such journals and $128$ such authors.  Journals in this two-mode core are American Sociological Review (with degree $122$), American Journal of Sociology ($112$), Social Forces ($90$), and Annual Review of Sociology ($85$). In Table \ref{ja_deg} those authors are listed that are linked to all $4$ journals in this two-mode core.
 
 \begin{table}[!h]\label{ja_deg}
 \caption{Authors in the $\core (85, 3; \deg_{\C}(v), \deg_{\C}(v))$ that are linked to all journals in this two-mode core.}
 \begin{tabular}{rlrlrl}
 1 & Breiger, R. & 10 & Kandel, D. & 19 & Olzak, S.   \\
 2 & Burt, R. &  11 & Keister, L. & 20 & Portes, A.   \\
 3 & DiMaggio, P. & 12 & Knoke, D. & 21 & Reskin, B.   \\
4 & Fischer, C. &  13 & Lieberson, S. &  22 & Ridgeway, C.  \\
5 & Friedkin, N. &  14 & Lin, N. &   23 & Sampson, R. \\
6 & Galaskiewicz, J. &  15 & Marsden, P. &  24 & Skvoretz, J. \\
7 & Glass, J. & 16 & McPherson, J.  &  25 & Thoits, P.   \\
8 &  Kalleberg, A. & 17 & Mizruchi, M. & \\
9 & Kalmijn, M.  &  18 & Nee, V.   &  
 \end{tabular}
 \end{table}

If we want to consider the number of works (the sum of weights of links) we can use 
$$f_B(v, \C) = \sum_{u \in N(v, \C)} w(v,u)$$
 and $g_B(v, \C) = \deg_{\C}(v)$ stays the same. In this generalized two-mode core with thresholds $p$ and $q$ are the journals that published at least $p$ works of authors in the core and the authors that published their works in at least $q$ journals in the core.

We determined generalized two-mode cores for these two functions. We used the algorithm for one fixed threshold. We selected $p \in \{0,1,2,5,10\}$ and determined the generalized two-mode cores for all five values of the fixed parameter $p$ and the non-fixed parameter $q$. Fig. \ref{sumdeg} presents  the diagram of a relation between the value of the parameter $q$ and the size of the generalized two-mode core. One can notice that the size of the generalized two-mode core is not the same for any pair of values  $(p, q).$ Because coordinates axes are in the logarithmic scale the point for the $\core (0,0;f_B, g_B)$ is not shown. Its size is equal to the size of the set of all nodes. In Fig. \ref{sumdegcore} the $\core(10,12;f_B,g_B)$ is shown for the two selected property functions. This is the smallest generalized two-mode core for $p=10$. In the $\core(10,12;f_B,g_B)$ are included all journals in which authors in the core published at least $10$ works in total; and all authors that each published his/her works in at least $12$ different journals in the core. The thickness of links represents the number of works an author published in a journal -- a thicker link means more works. One can notice that the journals data were not cleaned because the identification problem appears in Fig. \ref{sumdegcore}  -- the following pairs of journal identificators represent the same journal:
\begin{itemize}
\item \verb$amer sociol rev, am sociol rev:$ American Sociological Review,
\item \verb$amer j sociol, am j sociol:$ American Journal of Sociology,
\item \verb$adm sci q, admin sci quart:$ Administrative Science Quarterly.
\end{itemize}
Identificator \verb$sociol method$  represents one of the two journals: Sociological Methodology and Sociological Methods \& Research.%, and International Journal of Social Research Methodology. 
These two journals are present in Fig. \ref{sumdegcore} also with identificators \verb$sociol methodol$ and \verb$sociological$ \verb$methods$ respectively.

\begin{figure}[!ht]
\centering
\includegraphics[width=\textwidth]{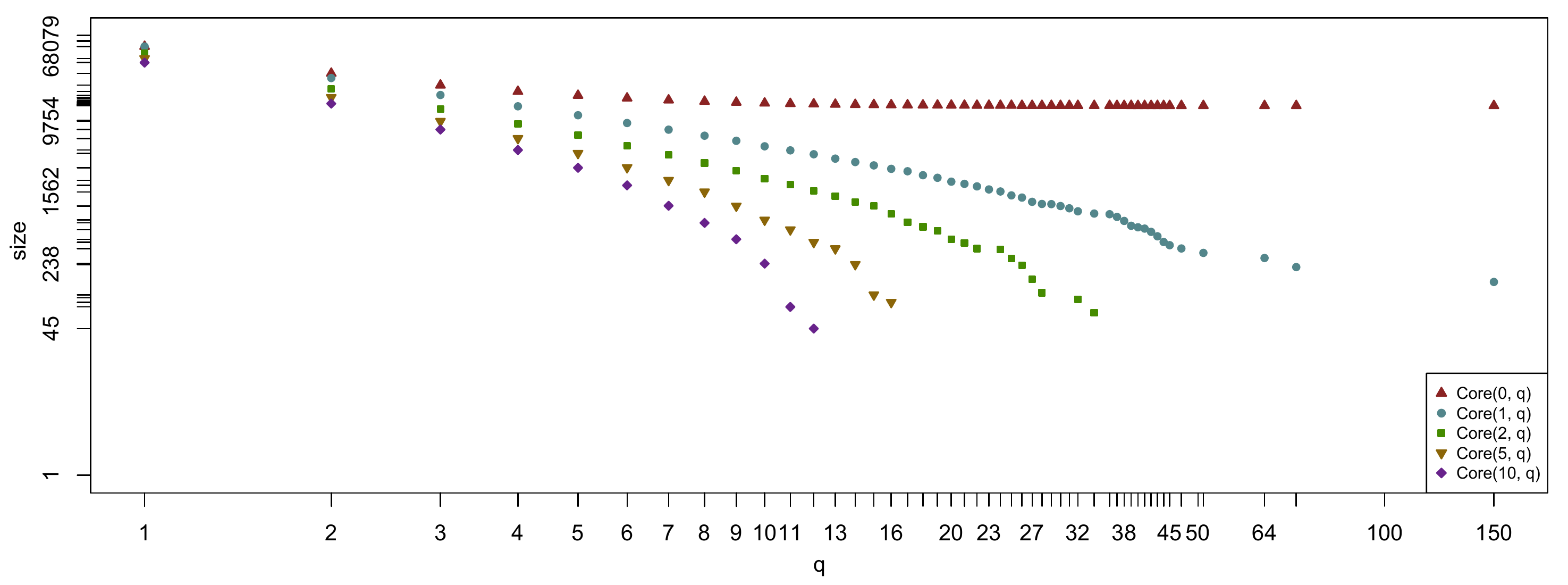}
\caption{A relation between the parameter $q$ and the size of the generalized two-mode core for the fixed values of parameter $p$ and the property functions $f_B$ and $g_B$.}
\label{sumdeg}
\end{figure}

\begin{figure}[!ht]
\centering
\includegraphics[width=\textwidth]{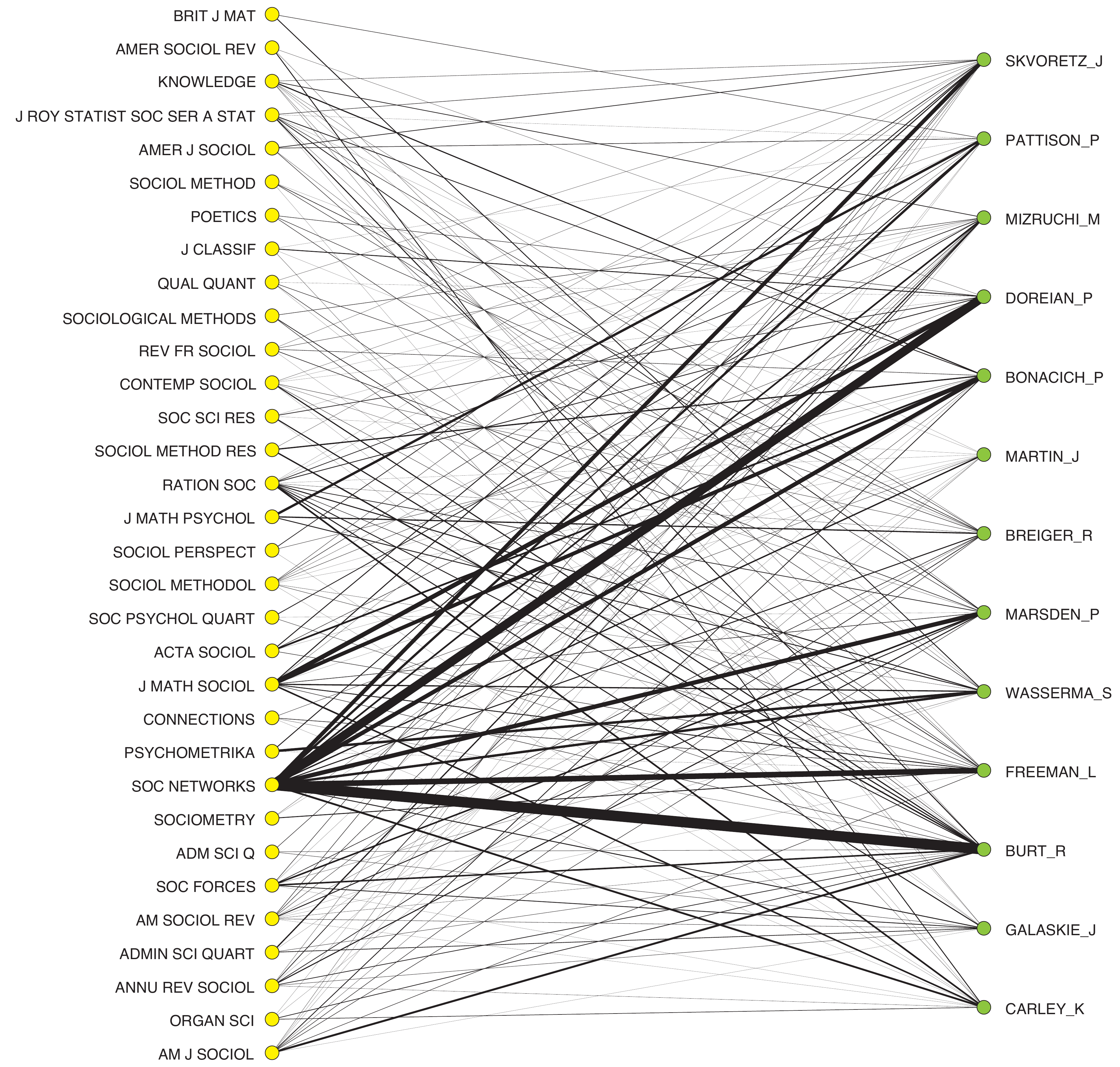}
\caption{A generalized two mode core for $p = 10$ and $q = 12$ and for the property functions $f_B$ and $g_B$.}
\label{sumdegcore}
\end{figure}

We could also select more complex property functions: 
$$f_C(v, \C) = \max_{u \in N(v, \C)} \deg(u) - \min_{u \in N(v, \C)} \deg(u)$$ and 
$$g_C(v, \C) = \frac{\sum_{u \in N(v, \C)} w(v,u)}{\sum_{u \in N(v)} w(v,u)}.$$ 
In the $\core (p,q; f_C,g_C)$ are the journals that published approximately (for a small value of $p$) the same number of works of each author in the core that is linked to those journals. And in this core are the authors that published at least $q \%$ of their works in journals that are in this core. This is one possible way to search for journals and authors that are tightly connected.

We determined the border of $(p,q)$ region for generalized two-mode cores for the property functions $f_C$ and $g_C$. The border is displayed in Fig. \ref{bordersn5}. At each boundary corner is written a pair of sizes of both sets of nodes in a generalized two-mode core. For example the $\core(70, 0.25; f_C,g_C)$ has $26$ nodes in the first set and $118$ nodes in the second set.

\begin{figure}[!ht]
\centering
\includegraphics[width=\textwidth]{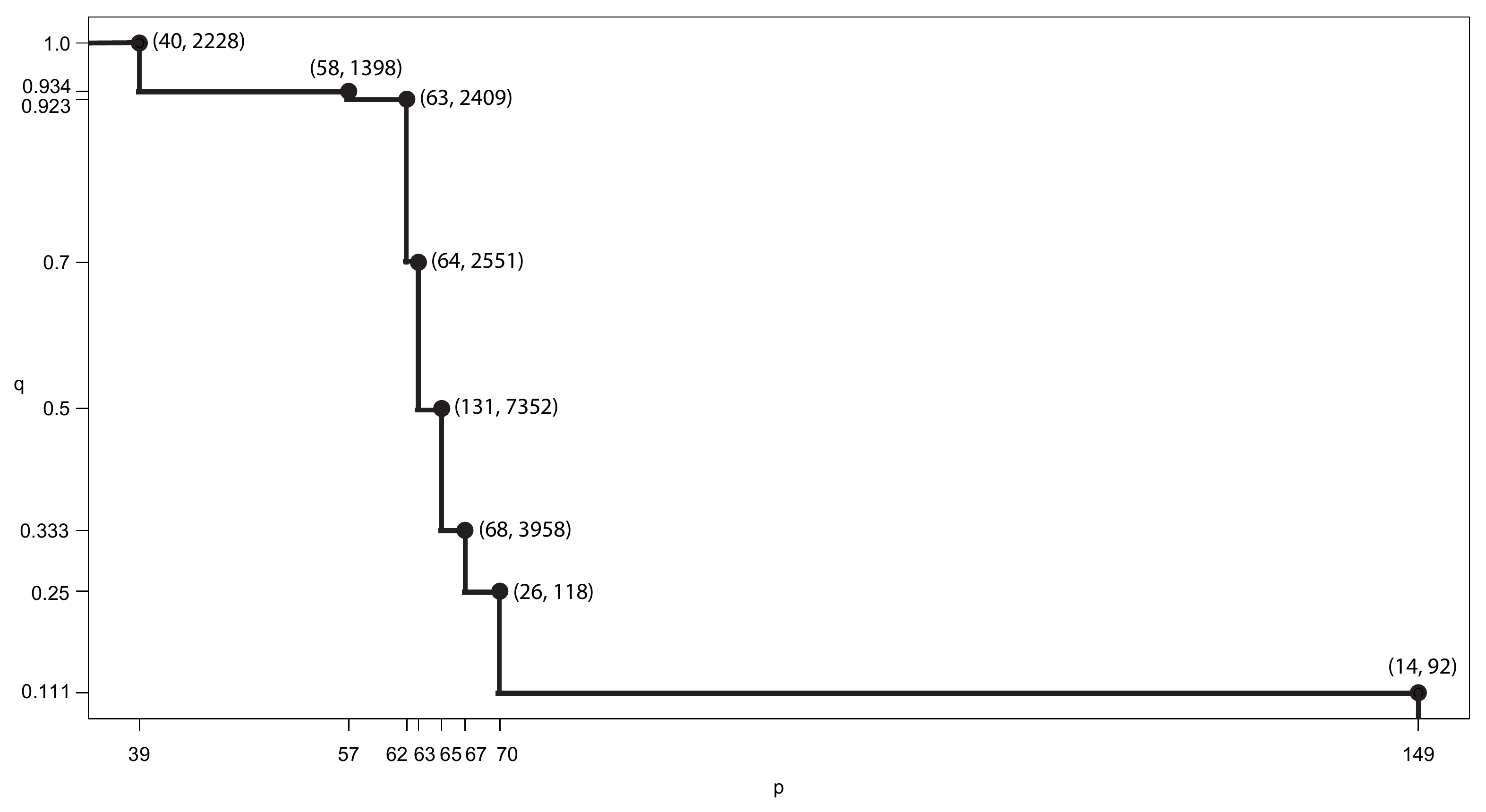}
\caption{The border of the $(p,q)$ region for $f_C$ and $g_C$.}
\label{bordersn5}
\end{figure}

\section{Conclusion}
\label{con}
In the paper we propose a new direct method for the analysis of two-mode networks. We provide an algorithm for determining generalized two-mode cores and present some examples of its application to real-life data. For the efficiency on large sparse networks we exploit the sparsity.  The presented approach can be straightforwardly extended to $r$-mode networks for $r>2$.

In our future work we intend to improve the efficiency of the algorithm and extend it for a use of the weights measured in nominal scales. We plan to make an experimental complexity analysis on random two-mode networks. For this task we also need to implement the generation of random two-mode networks of different types.

We already further elaborated Algorithm~\ref{alg} to produce the nested generalized two-mode cores for one fixed threshold that is shown in Algorithm~\ref{alg1}.
 We would like to improve this algorithm further -- to produce all $(p,q)$ pairs that determine different generalized two-mode cores and to identify only the boundary $(p, q)$ pairs as they are shown in Fig.~\ref{pq}. We  also intend to explore  the structure of the space of all generalized two-mode cores.
 
 An implementation of the proposed algorithms in Python is available at\\ http://zvonka.fmf.uni-lj.si/netbook/doku.php?id=pub:core.
 \\
 
 \textbf{Acknowledgements.}
The work was supported in part by the ARRS, Slovenia, grant J5-5537, as well as by grant within the EUROCORES Programme EUROGIGA (project GReGAS) of the European Science Foundation.

The first author was financed in part by the European Union, European Social Fund.


%% References


\bibliographystyle{elsarticle-num}

\bibliography{gtcores.bib}






% \end{thebibliography}

\end{document}